\documentclass[11pt,draft]{amsart}
\usepackage{mathrsfs}
\usepackage{amssymb}
\usepackage{amsmath}
\usepackage{amsfonts}
\usepackage{amsfonts,enumerate}
\usepackage{pifont}
\usepackage{enumerate}
\usepackage{graphics}
\usepackage{verbatim}
\usepackage{amsthm}
\numberwithin{equation}{section}
\newtheorem{theorem}{Theorem}[section]

\newtheorem{definition}{Definition}[section]
\newtheorem{lemma}{Lemma}[section]

\newtheorem{remark}{Remark}[section]

\newcommand{\8}{\infty}
\newcommand{\el}{\ell}

\newcommand{\be}{\begin{eqnarray*}}
\newcommand{\ee}{\end{eqnarray*}}
\newcommand{\beq}{\begin{equation}}
\newcommand{\eeq}{\end{equation}}
\newcommand{\beqn}{\begin{equation*}}
\newcommand{\eeqn}{\end{equation*}}
\newcommand{\bsp}{\begin{split}}
\newcommand{\esp}{\end{split}}

\begin{document}
\title{On the Cauchy problem for\\ Gross-Pitaevskii hierarchies}

\thanks{{\it 2010 Mathematics Subject Classification:} 35Q55, 81V70.}
\thanks{{\it Key words:} Gross-Pitaevskii hierarchy, nonlinear Schrodinger equation, Cauchy problem, Space-time type estimate.}

\author{Zeqian Chen}

\address{Wuhan Institute of Physics and Mathematics, Chinese Academy of Sciences, West District 30, Xiao-Hong-Shan, Wuhan 430071, China}



\author{Chuangye Liu}

\address{School of Mathematics and Statistics, Central China Normal University, Luo-Yu Road 152, Wuhan 430079, China}

\email{chuangyeliu1130@126.com}


\thanks{C.Liu is partially supported by NSFC grants No.11071095}

\date{}
\maketitle
\markboth{Z. Chen and C. Liu}%
{Gross-Pitaevskii hierarchies}

\begin{abstract}
The purpose of this paper is to investigate the Cauchy problem for the Gross-Pitaevskii infinite linear
hierarchy of equations on $\mathbb{R}^n,$ $n \geq 1.$ We prove local existence and uniqueness of solutions in certain Sobolev type spaces $\mathrm{H}^{\alpha}_{\xi}$ of sequences of marginal density operators with $\alpha > n/2.$ In particular, we give a clear discussion of all cases $\alpha > n/2,$ which covers the local well-posedness problem for Gross-Pitaevskii hierarchy in this situation.
\end{abstract}


\section{Introduction}\label{Intro}

Motivated by recent experimental realizations of Bose-Einstein condensation the theory of dilute, inhomogeneous Bose systems is
currently a subject of intensive studies in physics \cite{DGPS}. The ground state of bosonic atoms in a trap has been shown
experimentally to display Bose-Einstein condensation (BEC). This fact is proved theoretically by Lieb {\it et al} \cite{LS, LSY1,
LSY2} for bosons with two-body repulsive interaction potentials in the dilute limit, starting from the basic Schr\"{o}dinger equation.
On the other hand, it is well known that the dynamics of Bose-Einstein condensates are well described by the Gross-Pitaevskii
equation \cite{G1, G2, P}. A rigorous derivation of this equation from the basic many-body Schr\"{o}dinger equation in an appropriate
limit is not a simple matter, however, and has only achieved recently in three spatial dimensions by Elgart, Erd\"{o}s, Schlein
and Yau \cite{EESY, ESY1, ESY2, ESY3, ESYprl, EY}, based on the notion of so-called Gross-Pitaevskii hierarchies. In their program an important step is
to prove uniqueness to the Gross-Pitaevskii hierarchy via Feynman graph (see \cite{ESY2, KM}).

Recently, T.Chen and N.Pavlovi\'{c} \cite{CP1} started to investigate the Cauchy problem for the Gross-Pitaevskii hierarchy, using a Picard-type fixed point argument. In the present paper, we continue this line of investigation. We will prove local existence and uniqueness of solutions in certain Sobolev type spaces $\mathrm{H}^{\alpha}_{\xi}$ (for definition see Section \ref{PreMainresult} below) of sequences of marginal density operators with $\alpha > n/2.$ Instead of using a fixed point principle as in \cite{CP1}, here we use the fully expanded iterated Duhamel series, and a Cauchy convergence criterion, without additional conditions on any spacetime norms. The assumption of $\alpha> n/2$ allows us to significantly simplify the approaches and provide an improvement for the previous work \cite{CP1} in all cases $\alpha> n/2.$ Our proof involves the simple property that the interaction operators $B^{(k)}$ are bounded maps from the $k+1$-particle Hilbert space $\mathrm{H}^{\alpha}_{k+1}$ to the $k$-particle Hilbert space $\mathrm{H}^{\alpha}_k$ in the cubic case. A case of this type has previously been presented by Chen-Pavlovi\'{c} \cite{CP3} in their derivation of the quintic NLS for $n=1.$ In the much more difficult situation $\alpha \leq n/2,$ as done recently in \cite{CP2}, it is necessary to invoke the Strichartz estimates of the type introduced in the pioneering work of Klainerman-Machedon \cite{KM}.

The paper is organized as follows. In Section \ref{PreMainresult}, some notations and the main result are presented. Section \ref{PreEstimate} is
devoted to present elementary estimates which will be used later. In particular, we will prove the fact that the interaction operators $B^{(k)}$ are bounded maps from the $k+1$-particle Hilbert space $\mathrm{H}^{\alpha}_{k+1}$ to the $k$-particle Hilbert space $\mathrm{H}^{\alpha}_k$ in the cubic case for $\alpha > n/2.$ The proof is completely analogous to that of the classical Sobolev inequality $\| f \|_{L^{\8} (\mathbb{R}^n)} \le C \| f \|_{\mathrm{H}^{\alpha} (\mathbb{R}^n)}.$ In section \ref{ProofTh-nge1}, the main result is proved. Finally, in Section \ref{GPHierarchyQuintic}, we discuss the so-called quintic Gross-Pitaevskii hierarchy and extend the result obtained in the previous sections to that case.

\section{Preliminaries and statement of the main result}\label{PreMainresult}

\subsection{Gross-Pitaevskii hierarchies}

As follows, we denote by $x$ a general variable in $\mathbb{R}^n$ and by $\mathbf{x}=(x_1,\cdots , x_N)$ a point in $\mathbb{R}^{N n}.$ We will also use the notation $\mathbf{x}_k=(x_1,\ldots, x_k) \in \mathbb{R}^{k n}$ and $\mathbf{x}_{N-k}=(x_{k+1},\ldots, x_N)\in \mathbb{R}^{(N-k) n}.$ For a function $f$ on $\mathbb{R}^{k n}$ we
let
\be
(\Theta_{\sigma} f) (x_1, \ldots , x_k) = f (x_{\sigma (1)}, \ldots , x_{\sigma (k)} )
\ee
for any permutation $\sigma \in \Pi_k$ ($\Pi_k$ denotes the set of permutations on $k$ elements). Then, each $\Theta_{\sigma}$ is a unitary operator on $L^2 ( \mathbb{R}^{k n} ).$ A bounded operator $A$ on $L^2 ( \mathbb{R}^{k n} )$ is called {\it $k$-partite symmetric} or simply {\it symmetric} if
\beq\label{eq:OperatorSymm}
\Theta_{\sigma} A \Theta_{\sigma^{-1}} = A
\eeq
for every $\sigma \in \Pi_k.$ Evidently, a density operator $\gamma^{(k)}$ on $L^2 ( \mathbb{R}^{k n} )$ (i.e., $\gamma^{(k)} \ge 0$ and $\mathrm{tr} \gamma^{(k)} =1$) with the kernel function $\gamma^{(k)} (\mathbf{x}_k; \mathbf{x}'_k)$ is $k$-partite
symmetric if and only if\begin{equation*} \gamma^{(k)}
(x_1,\dotsc,x_k;x'_1,\dotsc,x'_k)= \gamma^{(k)}
(x_{\sigma(1)},\dotsc,x_{\sigma(k)};x'_{\sigma(1)},\dotsc,x'_{\sigma(k)})
\end{equation*}
for any $\sigma \in \Pi_k.$

Also, we set
\be
L^2_s ( \mathbb{R}^{k n} ) = \big \{ f \in L^2 ( \mathbb{R}^{k n} ):\; \Theta_{\sigma} f = f,\;\forall \sigma \in \Pi_k \big \},
\ee
equipped with the inner product of $L^2 ( \mathbb{R}^{k n} ).$ Clearly, $L^2_s ( \mathbb{R}^{k n} )$ is a Hilbert subspace of $L^2 ( \mathbb{R}^{k n} ).$ It is easy to check that any $k$-partite symmetric operator on $L^2 ( \mathbb{R}^{k n} )$ preserves $L^2_s ( \mathbb{R}^{k n} ).$

\begin{definition}\label{df:GPHierarchy}
Given $n \geq 1,$ the $n$-dimensional Gross-Pitaevskii (GP) hierarchy refers to a sequence $\{ \gamma^{(k)}_{t} \}_{k \geq 1}$ of $k$-partite symmetric density operators on $L^2 ( \mathbb{R}^{k n} ),$ where $t \geq 0,$ which satisfy the Gross-Pitaevskii infinite linear hierarchy of equations,
\beq\label{eq:GPHierarchyEqua}
i \partial_t \gamma^{(k)}_t = \big [- \Delta^{(k)}, \gamma^{(k)}_t \big ] + \mu B^{(k)} \gamma^{(k+1)}_t,~~\Delta^{(k)} = \sum^k_{j=1} \Delta_{x_j},\; \mu = \pm 1,
\eeq
with initial conditions
\be
\gamma^{(k)}_{t=0} = \gamma^{(k)}_{0},\quad k =1, 2, \ldots.
\ee
Here, $\Delta_{x_j}$ refers to the usual Laplace operator with respect to the variables $x_j \in {\mathbb R}^n$ and the operator $B^{(k)}$ is defined by
\be
B^{(k)} \gamma^{(k+1)}_t = \sum^k_{j=1} \mathrm{t r}_{k+1} \big [\delta(x_j-x_{k+1}),\gamma^{(k+1)}_t \big ]
\ee
where the notation $\mathrm{t r}_{k+1}$ indicates that the trace is taken over the $(k+1)$-th variable.
\end{definition}

As in \cite{CP1}, we refer to \eqref{eq:GPHierarchyEqua} as the cubic GP hierarchy. For $\mu= 1$ or $\mu = -1$ we refer to the corresponding
GP hierarchies as being defocusing or focusing, respectively. We note that the cubic Gross--Pitaevskii hierarchy accounts for two-body interactions between the Bose particles (e.g., see \cite{DGPS, ESYprl} and references therein for details).

\begin{remark}\label{rk:GPHierarchyEquaFunct}
In terms of the kernel functions $\gamma^{(k)}_{t}({\bf x}_k;{\bf x}'_k),$ we can rewrite \eqref{eq:GPHierarchyEqua} as follows:
\beq\label{eq:GPHierarchyEquaFunct}
\big ( i \partial_t + \triangle^{(k)}_{\pm} \big ) \gamma^{(k)}_t ({\bf x}_k;{\bf x}'_k) = \mu \big [ B^{(k)} \gamma^{(k+1)}_t \big ] ({\bf x}_k; {\bf x}'_k ),
\eeq
where $\triangle^{(k)}_{\pm} = \sum^k_{j=1} ( \Delta_{x_j} - \Delta_{x'_j} ),$ with initial conditions
\be
\gamma^{(k)}_{t=0}({\bf x}_k;{\bf x}'_k)=\gamma^{(k)}_{0}({\bf x}_k;{\bf x}'_k),\quad k =1, 2, \ldots .
\ee
In particular, the action of $B^{(k)}$ on density operators with smooth kernel functions, $\gamma^{(k+1)} ({\bf x}_{k+1}; {\bf x}'_{k+1})
\in \mathcal{S} (\mathbb{R}^{(k+1)n} \times \mathbb{R}^{(k+1)n}),$ is given by
\beq\label{eq:BOperatorFunct}
\begin{split}
\big [ B^{(k)} \gamma^{(k+1)}_t \big ] ( {\bf x}_{k}; {\bf x}'_{k} )=& \sum^k_{j=1} \int d x_{k+1} d x'_{k+1} \gamma^{(k+1)}_t ({\bf x}_{k}, x_{k+1}; {\bf x}'_{k}, x'_{k+1})\\
& \; \times \delta ( x'_{k+1} - x_{k+1} ) \big [ \delta(x_j-x_{k+1}) - \delta(x'_j-x_{k+1}) \big ].
\end{split}
\eeq
The action of $B^{(k)}$ can be extended to generic density operators. This will be made precise in Lemma \ref{le:BOperatorEstimate-nge1}.
\end{remark}

\begin{remark}\label{rk:GPEqua}
Let $\varphi \in \mathrm{H}^1(\mathbb{R}^n),$ then one can easily verify that a particular solution to \eqref{eq:GPHierarchyEquaFunct} with initial conditions
\be
\gamma^{(k)}_{t=0}({\bf x}_k; {\bf x}'_k) = \prod^k_{j=1} \varphi(x_j) \overline{\varphi(x'_j)},\quad k=1,2,\ldots,
\ee
is given by
\be
\gamma^{(k)}_t ({\bf x}_{k}; {\bf x}'_{k} ) = \prod^k_{j=1} \varphi_t (x_j) \overline{\varphi_t ( x'_j )}~~k=1,2,\ldots,
\ee
where $\varphi_t$ satisfies the cubic non-linear Schr\"odinger equation
\beq\label{eq:GPEqua}
i\partial_t \varphi_t = -\Delta \varphi_t + \mu |\varphi_t|^2 \varphi_t,\quad \varphi_{t=0}=\varphi,
\eeq
which is {\it defocusing} if $\mu =1,$ and {\it focusing} if $\mu= -1.$
\end{remark}

The Gross-Pitaevskii hierarchy \eqref{eq:GPHierarchyEqua} can be written in the integral form
\beq\label{eq:GPHierarchyIntEqua}
\gamma^{(k)}_t = {\mathcal U}^{(k)}_0(t) \gamma^{(k)}_0 + \int^{t}_{0} d s~ {\mathcal U}^{(k)}_0 (t-s) \tilde{B}^{(k)} \gamma^{(k+1)}_s,\; k=1,2,\ldots,
\eeq
where $\tilde{B}^{(k)} = - i \mu B^{(k)}.$ Hereafter, the free evolution operator is defined by
\be
{\mathcal U}^{(k)}_0(t) A =\exp \big ( i t \Delta^{(k)} \big ) A \exp \big(- it \Delta^{(k)} \big ),~~k=1,2,\ldots,
\ee
for every operator $A$ on $L^2 ( \mathbb{R}^{k n} ).$ The action of ${\mathcal U}^{(k)}_0(t)$ on kernel functions $\gamma^{(k)} \in L^2 ( \mathbb{R}^{k n} \times \mathbb{R}^{k n} )$ is given by
\beq\label{eq:UActionFunct}
{\mathcal U}^{(k)}_0(t) \gamma^{(k)} ({\bf x}_k, {\bf x}'_k ) = e^{ - i t \triangle^{(k)}_{\pm}} \gamma^{(k)} ({\bf x}_k, {\bf x}'_k ).
\eeq
Formally we can expand the solution $\gamma^{(k)}_t$ of \eqref{eq:GPHierarchyIntEqua}
for any $m \geq 1$ as
\beq\label{eq:DuhamelExpan}
\begin{split}
\gamma^{(k)}_t & = {\mathcal U}^{(k)}_0 (t) \gamma^{(k)}_0 + \sum^{m-1}_{j=1} \int^t_0 d s_1 \int^{s_1}_0 d s_2 \cdots
\int^{s_{j-1}}_0 d s_j {\mathcal U}^{(k)}_0(t-s_1) \tilde{B}^{(k)} \cdots\\
& \; \times {\mathcal U}^{(k+j-1)}_0 ( s_{j-1} - s_j ) \tilde{B}^{(k+j-1)} {\mathcal U}^{(k+j)}_0 (s_j) \gamma^{(k+j)}_0\\
& \; + \int^t_0 d s_1 \int^{s_1}_0 d s_2 \cdots \int^{s_{m-1}}_0 d s_m {\mathcal U}^{(k)}_0 ( t-s_1 ) \tilde{B}^{(k)} \cdots \\
& \; \times {\mathcal U}^{(k+m-1)}_0 ( s_{m-1} - s_m ) \tilde{B}^{(k+m-1)} \gamma^{(k+m)}_{s_m},
\end{split}
\eeq
with the convention $s_0 =t.$ The terms in the summation contain only the initial data. The last error term involves the density operator at an intermediate time $s_m.$

\subsection{Statement of the main result}

In order to state our main results, we require some more notation. We will use $\gamma^{(k)}, \rho^{(k)}$ for denoting either (density) operators or kernel functions. For $k \geq 1$ and $\alpha >0,$ we denote by $\mathrm{H}^{\alpha}_k = \mathrm{H}^{\alpha} (\mathbb{R}^{ k n} \times \mathbb{R}^{k n})$ the space of measurable functions $\gamma^{(k)} = \gamma^{(k)} ( {\bf x}_k, {\bf x}'_k )$ in $L^2(\mathbb{R}^{k n} \times \mathbb{R}^{k n})$ such that
\be\begin{split}
\| \gamma^{(k)} \|_{\mathrm{H}^{\alpha}_k} : = \| S^{(k, \alpha)} \gamma^{(k)} \|_{L^2(\mathbb{R}^{k n} \times \mathbb{R}^{k n})} < \8,
\end{split}\ee
where
\be\begin{split}
S^{(k,\alpha)} :  = \prod_{j=1}^k \big [ (1 - \Delta_{x_j} )^{\frac{\alpha}{2}} (1 - \Delta_{x'_j} )^{\frac{\alpha}{2}} \big ].
\end{split}\ee
Evidently, $\mathrm{H}^{\alpha}_k$ is a Hilbert space with the inner product
\be\begin{split}
\langle \gamma^{(k)}, \rho^{(k)} \rangle = \big \langle S^{(k,\alpha)} \gamma^{(k)},\; S^{(k,\alpha)} \rho^{(k)} \big \rangle_{L^2(\mathbb{R}^{k n} \times \mathbb{R}^{k n})}.
\end{split}\ee
Moreover, the norm $\|\cdot\|_{\mathrm{H}^{\alpha}_k}$ is invariance under the action of $\mathcal{U}_0^{(k)}(t),$ that is,
\be
\| \mathcal{U}_0^{(k)}(t) \gamma^{(k)} \|_{\mathrm{H}^{\alpha}_k} = \| \gamma^{(k)} \|_{\mathrm{H}^{\alpha}_k}
\ee
because $\exp \{\pm it \Delta^{(k)} \}$ commutate with $\Delta_{x_j}$ for any $j.$

Let $0 < \xi < 1$ and $\alpha >0,$ we define
\beq\label{eq:SobolevSpace}
\mathcal{H}_{\xi}^{\alpha} = \Big \{ \Gamma = \{ \gamma^{(k)} \}_{k \ge 1} \in \bigotimes_{k=1}^{\infty} \mathrm{H}^{\alpha}_k :\; \|\Gamma\|_{\mathcal{H}_{\xi}^{\alpha}} : = \sum_{k=1}^{\infty}\xi^k\|\gamma^{(k)}\|_{\mathrm{H}^{\alpha}_k} < \infty \Big \}.
\end{equation}
Evidently, $\mathcal{H}^{\alpha}_{\xi}$ is a Banach space equipped with the norm $\| \cdot \|_{\mathcal{H}_{\xi}^{\alpha}},$ which is introduced in \cite{CP1}. We remark that similar spaces are used in the isospectral renormalization group analysis of spectral problems in quantum field theory (see \cite{BCFS}).

\begin{definition}\label{df:MildSolution}
For $T >0,$ $\Gamma_t = \{ \gamma^{(k)}_t \}_{k \geq 1} \in C([0,T], \mathcal{H}_{\xi}^{\alpha})$ is said to be a local $(mild)$ solution
to the Gross-Pitaevskii hierarchy \eqref{eq:GPHierarchyEqua} if for every $k=1,2,\ldots,$
\be
\gamma^{(k)}_t = {\mathcal U}^{(k)}_0(t) \gamma^{(k)}_0 + \int^{t}_{0} d s\, {\mathcal U}^{(k)}_0 (t-s) \tilde{B}^{(k)} \gamma^{(k+1)}_s,\quad \forall t \in [0,T],
\ee
holds in $\mathrm{H}^{\alpha}_k.$

\end{definition}

Our main result in this paper is the following theorem.

\begin{theorem}\label{th:nge1}
Assume that $n \geq 1$ and $\alpha > n/2.$ Suppose $\Gamma_0 = \{\gamma^{(k)}_{0}\}_{k \geq 1}\in\mathcal{H}_{\xi}^{\alpha}$ for some $0<\xi<1.$ Then there exists a constant $C = C_{\alpha, n}$ depending only on $n$ and $\alpha$ such that, for a fixed $0 < T < \xi/C$ with $\eta = \xi - C T,$ the following hold.
\begin{enumerate}[{\rm (i)}]

\item There exists a solution $\Gamma_t = \{\gamma^{(k)}_t \}_{k \geq 1} \in C ( [0,T], \mathcal{H}_{\eta}^{\alpha})$ to the Gross-Pitaevskii hierarchy \eqref{eq:GPHierarchyEqua} with the initial data $\Gamma_0$ satisfying
\begin{equation}\label{eq:SpacetimeEstimate-nge1}
\| \Gamma_t \|_{C( [0,T],\mathcal{H}_{\eta}^{\alpha})} \leq \frac{\eta}{\xi}\| \Gamma_0 \|_{\mathcal{H}_{\xi}^{\alpha}}.
\end{equation}

\item For $T = \xi/(5 C),$ if $\Gamma_t$ and $\Gamma'_t$ in $C ( [0,T], \mathcal{H}_{\eta}^{\alpha})$ are two solutions to \eqref{eq:GPHierarchyEqua} with initial conditions $\Gamma_{t=0} = \Gamma_0$ and $\Gamma'_{t=0} = \Gamma'_0$ in $\mathcal{H}^{\alpha}_{\xi}$ respectively, then
\begin{equation}\label{eq:SpacetimeEstimate-Twosolution}
\| \Gamma_t - \Gamma'_t \|_{C( [0,T], \mathcal{H}_{\eta}^{\alpha})} \leq \frac{4}{5} \| \Gamma_0 - \Gamma'_0  \|_{\mathcal{H}_{\xi}^{\alpha}}.
\end{equation}
Consequently, the solution $\Gamma_t$ to the initial problem \eqref{eq:GPHierarchyEqua} with the initial data in $\mathcal{H}^{\alpha}_{\xi}$ is unique in $C( [0,T], \mathcal{H}_{\eta}^{\alpha})$ for any $0 < T < \xi/ C.$

\end{enumerate}
\end{theorem}

\begin{remark}\label{rk:mainresult}
We will prove this theorem by the method of infinitely iterating the Duhamel series, and proving Cauchy convergence without additional conditions on spacestime bounds, and however, our argument won't work if $\alpha \le n/2.$ The assumption of $\alpha> n/2$ allows us to significantly simplify the approaches in the previous work \cite{CP1} and improve the corresponding one. In the much more difficult situation $\alpha\leq  \frac{n}{2},$ as done recently in \cite{CP2}, it is necessary to invoke the Strichartz estimates of the type introduced in the pioneering work of Klainerman-Machedon \cite{KM}.
\end{remark}

\section{Preliminary estimates}\label{PreEstimate}

In the sequel, we will mostly work in Fourier (momentum) space. Following \cite{ESY2}, we use the convention that variables $p,q,r, p', q', r'$
always refer to $n$ dimensional Fourier variables, while $x, x', y, y', z, z'$ denote the position space variables. With this convention, the usual hat indicating the Fourier transform will be omitted. For example, for $k \geq 1$ the kernel of a bounded operator $A$ on $L^2 ( \mathbb{R}^{k n} )$ in position space is $K (
\mathbf{x}_k ; \mathbf{x}'_k),$ then in the momentum space it is given by the Fourier transform
\be
K ( \mathbf{q}_k ; \mathbf{q}'_k) = \big \langle K, e^{-i \langle \cdot, \mathbf{q}_k \rangle}e^{ i \langle \cdot, \mathbf{q}'_k \rangle} \big \rangle = \int d \mathbf{x}_k d \mathbf{x}'_k K ( \mathbf{x}_k ; \mathbf{x}'_k ) e^{-i \langle \mathbf{x}_k, \mathbf{q}_k \rangle}e^{ i \langle \mathbf{x}'_k, \mathbf{q}'_k \rangle},
\ee
with the slight abuse of notation of omitting the hat on left hand side. Here,
\be
\langle \mathbf{x}_k, \mathbf{q}_k \rangle = \sum^k_{j=1} x_j \cdot q_j, \quad \forall \mathbf{x}_k = (x_1, \ldots , x_k), \mathbf{q}_k = (q_1, \ldots , q_k) \in \mathbb{R}^{k n}.
\ee


Thus, on kernels in the momentum space $B^{(k)}$ in \eqref{eq:BOperatorFunct} acts according to
\beq\label{eq:BOperatorFunctMomentum}
\begin{split}
\big [ B^{(k)}& \gamma^{(k+1)} \big ] ({\bf p}_k;{\bf p}'_k)\\
= & \sum^k_{j=1} \int d q_{k+1} d q'_{k+1}\\
&\;\; \times \Big \{ \gamma^{(k+1)}(p_1,\dotsc, p_j-q_{k+1}+q'_{k+1}, \dotsc, p_k,q_{k+1}; {\mathtt p}'_k, q'_{k+1})\\
& \quad - \gamma^{(k+1)} ({\bf p}_k, q_{k+1}; p'_1, \dotsc, p'_j + q_{k+1} - q'_{k+1}, \dotsc, p'_k, q'_{k+1} ) \Big \}\\
= & \sum^k_{j=1} \int d {\bf q}_{k+1} d {\bf q}'_{k+1} \Big [ \prod^k_{l \neq j} \delta(p_l - q_l) \delta (p'_l-q'_l) \Big ]\\
& \;\;\times \gamma^{(k+1)} ( {\bf q}_{k+1}; {\bf q}'_{k+1}) \Big \{ \delta(p'_j-q'_j) \delta \big ( p_j- [q_j+q_{k+1}-q'_{k+1} ] \big )\\
& \quad - \delta (p_j - q_j) \delta \big ( p'_j - [ q'_j+q'_{k+1}-q_{k+1} ] \big ) \Big \}.
\end{split}
\eeq


We begin with the following simple lemma.

\begin{lemma}\label{le:IntIneuqa-nge1}
If $\beta >n,$ then
\begin{equation}\label{eq:IntInequa-nge1}
\sup_{p \in \mathbb{R}^n} \int_{\mathbb{R}^n} \mathrm{d} q \mathrm{d} q' \dfrac{\langle p \rangle^{\beta}}{ \langle p + q' -q \rangle^{\beta} \langle q \rangle^{\beta} \langle q' \rangle^{\beta}}<\infty.
\end{equation}
\end{lemma}

\begin{proof}
Since
\be\begin{split}
& \dfrac1{\langle p + q'-q \rangle^{\beta} \langle q \rangle^{\beta} \langle q' \rangle^{\beta}}\\
\leq & \dfrac{2^{\beta}}{\langle p + q' \rangle^{\beta} \langle q' \rangle^{\beta}} \Big( \dfrac1{\langle p + q' -q \rangle^{\beta}} + \dfrac1{ \langle q \rangle^{\beta }} \Big )\\
\le & \dfrac{2^{2 \beta}}{\langle p \rangle^{\beta} }\Big( \dfrac1{\langle p + q'\rangle^{\beta}} + \dfrac1{ \langle q' \rangle^{\beta }} \Big ) \Big( \dfrac1{\langle p + q' -q \rangle^{\beta}} + \dfrac1{ \langle q \rangle^{\beta }} \Big ),
\end{split}\ee
the inequality \eqref{eq:IntInequa-nge1} is concluded from the assumption $\beta > n.$
\end{proof}

As in \cite{KM}, we introduce for $\gamma^{(k+1)} ({\bf x}_{k+1},{\bf x}'_{k+1}) \in \mathcal{S} (\mathbb{R}^{(k+1) n} \times \mathbb{R}^{(k+1) n}),$
\begin{equation*}
\begin{split}
[ B_{j,k}^1 & \gamma^{(k+1)}]  ({\bf x}_{k},{\bf x}'_{k})\\
= & \int \mathrm{d} x_{k+1} \mathrm{d}x'_{k+1}
\delta(x_{k+1}-x'_{k+1}) \delta(x_j- x_{k+1}) \gamma^{(k+1)}({\bf
x}_{k+1},{\bf x}'_{k+1}),
\end{split}
\end{equation*}
and
\begin{equation*}
\begin{split}
[B_{j,k}^2 & \gamma^{(k+1)}] ({\bf x}_{k},{\bf x}'_{k})\\
= & \int \mathrm{d} x_{k+1} \mathrm{d} x'_{k+1}
\delta(x_{k+1}-x'_{k+1} ) \delta(x'_j - x_{k+1})\gamma^{(k+1)}({\bf
x}_{k+1},{\bf x}'_{k+1}),
\end{split}
\end{equation*}
where $j=1,\ldots, k.$ Then, by \eqref{eq:BOperatorFunct} we have
\be
B^{(k)} = \sum_{j=1}^k \,\big ( B_{j,k}^1 - B_{j,k}^2 \big )
\ee
acting on smooth kernel functions $\gamma^{(k+1)} \in \mathcal{S} (\mathbb{R}^{(k+1) n} \times \mathbb{R}^{(k+1) n}).$

The following estimate is crucial for the proof of Theorem \ref{th:nge1}. In fact, the proof of the lemma is completely analogous to that of $\|f\|_{L^{\infty}(\mathbb{R}^n)}\leq  C\|f\|_{H^{\alpha}(\mathbb{R}^n) }$  based on Fourier analysis.

\begin{lemma}\label{le:BOperatorEstimate-nge1}
Suppose that $\alpha>\dfrac{n}{2}$ and $n\geq 1.$ Then, there exists a constant $ C_{\alpha,n} >0$ depending only on $\alpha$ and $n$ such that, for any $\gamma^{(k+1)} \in \mathcal{S} (\mathbb{R}^{(k+1) n} \times \mathbb{R}^{(k+1) n}),$
\be
\| B_{j,k}^l \gamma^{(k+1)} \|_{\mathrm{H}^{\alpha}_k} \leq C_{\alpha,n} \|\gamma^{(k+1)} \|_{\mathrm{H}^{\alpha}_{k+1}},\quad l=1,2,
\ee
for all $k \geq 1,$ where $j=1,\cdots,k.$ Consequently,
\beq\label{eq:BoperatorEstimate-nge1}
\|B^{(k)} \gamma^{(k+1)} \|_{\mathrm{H}^{\alpha}_k} \leq C_{\alpha,n} k \|\gamma^{(k+1)}\|_{\mathrm{H}^{\alpha}_{k+1}}
\eeq
for any $\gamma^{(k+1)} \in \mathcal{S} (\mathbb{R}^{(k+1) n} \times \mathbb{R}^{(k+1) n})$ and all $k\geq 1.$
\end{lemma}

\begin{remark}\label{rk:BOperatorDf}
The estimate \eqref{eq:BoperatorEstimate-nge1} indicates that the operator $B^{(k)},$ originally defined on Schwarz functions, can be extended to a
bounded operator from $\mathrm{H}^{\alpha}_{k+1}$ to $\mathrm{H}^{\alpha}_k.$ In this case, we still denote it by $B^{(k)}.$
\end{remark}

\begin{proof}
We first consider $B^1_{1,k}.$ For $\gamma^{(k+1)} \in \mathcal{S} (\mathbb{R}^{(k+1) n} \times \mathbb{R}^{(k+1) n})$, from the Plancherel's theorem it is concluded that
\begin{equation*}
\begin{split}
\|B_{1,k}^1 \gamma^{(k+1)}\|_{\mathrm{H}^{\alpha}_k}^2 = & \int \prod_{j=1}^k\langle p_j\rangle^{2 \alpha}\langle p'_j \rangle^{2
\alpha} \mathrm{d}{\bf p}_k \mathrm{d} {\bf p}'_k \Big | \int \mathrm{d}q\mathrm{d}q'\\
& \gamma^{(k+1)}(p_1+q'-q,p_2, \cdots,p_k,q;p'_1,\cdots,p'_k,q') \Big|^2.
\end{split}
\end{equation*}
Then, we obtain
\begin{equation*}
\begin{split}
\|B_{1,k}^1 & \gamma^{(k+1)}\|_{\mathrm{H}^{\alpha}_k}^2\\
\leq & \int \prod_{j=1}^k\langle p_j\rangle^{2\alpha}\langle p'_j\rangle^{2\alpha} \mathrm{d}{\bf p}_k\mathrm{d}{\bf p}'_k\\
& \times \Big ( \int \mathrm{d} q \mathrm{d} q' \dfrac{1}{\langle p_1+q'-q\rangle^{2\alpha}\langle q\rangle^{2\alpha}\langle q' \rangle^{2\alpha}} \Big )\\
& \times \Big ( \int \mathrm{d}q\mathrm{d}q'\langle p_1+q'-q\rangle^{2\alpha}\langle q\rangle^{2\alpha}\langle q'\rangle^{2\alpha}\\
& \times | \gamma^{(k+1)} (p_1+q'-q,p_2, \cdots, p_k,q;p'_1,\cdots,p'_k,q')|^2 \Big )\\
\leq & \sup_{p_1 \in \mathbb{R}^n} \int \mathrm{d} q \mathrm{d}q' \dfrac{\langle p_1 \rangle^{2 \alpha}}{\langle p_1 + q'-q \rangle^{2\alpha} \langle q \rangle^{2\alpha} \langle q'\rangle^{2\alpha}}\\
& \times \int \prod_{j = 2}^k \langle p_j \rangle^{2\alpha} \langle p'_j \rangle^{2 \alpha} \mathrm{d} {\bf p}_k\mathrm{d} {\bf p}'_k \mathrm{d} q \mathrm{d} q'\\
& \times \Big \{ \langle p_1 +q'-q \rangle^{2\alpha} \langle p'_1 \rangle^{2 \alpha} \langle q \rangle^{2\alpha}\langle q'\rangle^{2\alpha}\\
& \times | \gamma^{(k+1)} (p_1+q'-q, p_2, \cdots, p_k,q; p'_1,\cdots,p'_k,q') |^2 \Big \}\\
= & \sup_{p_1 \in \mathbb{R}^n} \int \mathrm{d} q \mathrm{d} q' \dfrac{\langle p_1 \rangle^{2\alpha}}{\langle p_1 + q' - q \rangle^{2\alpha} \langle q \rangle^{2 \alpha} \langle q'\rangle^{2\alpha}} \| \gamma^{(k+1)}\|_{\mathrm{H}^{\alpha}_{k+1}}^2\\
\leq & C_{\alpha, n} \| \gamma^{(k+1)} \|^2_{\mathrm{H}^{\alpha}_{k+1}},
\end{split}
\end{equation*}
where we have used Lemma \ref{le:IntIneuqa-nge1} in the last inequality. For the operator $B_{1,k}^2,$ we have the same estimate
\begin{equation*}
\|B_{1,k}^2\gamma^{(k+1)}\|_{\mathrm{H}^{\alpha}_k} \leq C_{\alpha,n} \|\gamma^{(k+1)}\|_{\mathrm{H}^{\alpha}_{k+1}}.
\end{equation*}
Similarly, we can prove the same bound for $B_{j,k}^1$ and $B_{j,k}^2$ when $j=2,\cdots, k.$ Consequently, we conclude the estimate \eqref{eq:BoperatorEstimate-nge1}.
\end{proof}

\section{Proof of Theorem \ref{th:nge1}}\label{ProofTh-nge1}

Now we are ready to prove Theorem \ref{th:nge1}. The proof is divided into two parts as follows.

\begin{proof}
(i)\; Let $\alpha > n/2$ and $0< \xi < 1.$ Given $\Gamma_0 = \{ \gamma^{(k)}_0 \}_{k \ge 1} \in \mathcal{H}^{\alpha}_{\xi}.$ For $m \geq 1,$ set
\begin{equation}\label{eq:m-iterate}
\gamma^{(k)}_{m,t} = {\mathcal U}^{(k)}_0 (t) \gamma^{(k)}_0 + \int^{t}_{0} \mathrm{d} s\, {\mathcal U}^{(k)}_0 (t-s) \tilde{B}^{(k)} \gamma^{(k+1)}_{m-1,s},\quad t>0, k \geq 1,
\end{equation}
with the convention $\gamma^{(k)}_{0,t} \equiv \gamma^{(k)}_{0},$ where $\tilde{B}^{(k)} = -i \mu B^{(k)}$ (e.g., \eqref{eq:GPHierarchyIntEqua}). By expansion, for every $m \geq 1$ one has
\begin{equation*}
\begin{split}
\gamma^{(k)}_{m,t}= & {\mathcal U}^{(k)}_0(t) \gamma^{(k)}_{0} +
\sum_{j=1}^{m-1}\int_{0}^{t}\mathrm{d}t_1\int_0^{t_1}\mathrm{d}t_2
\cdots \int_0^{t_{j-1}}\mathrm{d}
t_j {\mathcal U}^{(k)}_0(t-t_1) \tilde{B}^{(k)} \cdots\\
& \times {\mathcal U}^{(k+j-1)}_0(t_{j-1}-t_j) \tilde{B}^{(k+j-1)}
{\mathcal U}^{(k+j)}_0(t_j) \gamma^{(k+j)}_0\\
& + \int_{0}^{t} \mathrm{d} t_1 \int_0^{t_1} \mathrm{d} t_2 \cdots
\int_0^{t_{m-1}} \mathrm{d} t_m {\mathcal U}^{(k)}_0 (t-t_1)
\tilde{B}^{(k)} \cdots\\
&\times {\mathcal U}^{(k+m-1)}_0 (t_{m-1}-t_m) \tilde{B}^{(k+m-1)}
\gamma^{(k+m)}_{0}\\
\triangleq & \sum_{j=0}^{m} \Xi_{j,t}^{(k)},
\end{split}
\end{equation*}
with the convention $t_0 =t.$ Then, for $j=1,\cdots,m-1$, by Lemma \ref{le:BOperatorEstimate-nge1} we
have
\beq\label{eq:XiEsitmate}
\begin{split}
\| \Xi_{j,t}^{(k)} \|_{\mathrm{H}^{\alpha}_k} \leq & \int_0^t \mathrm{d} t_1 \int_0^{t_1} \mathrm{d} t_2 \cdots \int_0^{t_{j-1}}
\mathrm{d} t_j \Big \|{\mathcal U}^{(k)}_0(t-t_1) \tilde{B}^{(k)} \cdots \\
& \times {\mathcal U}^{(k+j-1)}_0(t_{j-1}-t_j) \tilde{B}^{(k+j-1)} {\mathcal U}^{(k+j)}_0(t_j) \gamma^{(k+j)}_{0} \Big \|_{\mathrm{H}^{\alpha}_k}\\
\leq & \int_{0}^{t} \mathrm{d} t_1 \int_0^{t_1} \mathrm{d}t_2 \cdots \int_0^{t_{j-1}} \mathrm{d} t_j k \cdots \\
& \times (k+j-1)(C_{\alpha,n})^{j} \| {\mathcal U}^{(k+j)}_0(t_j)\gamma^{(k+j)}_{0}\|_{\mathrm{H}^{\alpha}_{k+j}}\\
= & \dfrac{t^j}{j!} k \cdots (k+j-1) (C_{\alpha,n})^j \|\gamma^{(k+j)}_{0}\|_{\mathrm{H}^{\alpha}_{k+j}}\\
= & \binom{k+j-1}{j} (C_{\alpha,n}  t )^j\|\gamma^{(k+j)}_{0}\|_{\mathrm{H}^{\alpha}_{k+j}},
\end{split}
\eeq
and
\begin{equation*}
\begin{split}
\| \Xi_{m,t}^{(k)}\|_{\mathrm{H}^{\alpha}_k} \leq & \int_{0}^{t} \mathrm{d} t_1 \int_0^{t_1} \mathrm{d} t_2 \cdots \int_0^{t_{m-1}} \mathrm{d} t_m \Big \| {\mathcal U}^{(k)}_0(t-t_1) \tilde{B}^{(k)} \cdots \\
& \times {\mathcal U}^{(k+m-1 )}_0(t_{m-1}-t_m) \tilde{B}^{(k+m-1)} \gamma^{(k+m)}_0 \Big \|_{\mathrm{H}^{\alpha}_k}\\
\leq & \int_{0}^{t}\mathrm{d}t_1\int_0^{t_1}\mathrm{d}t_2\cdots\int_0^{t_{m-1}}\mathrm{d} t_m~ k \cdots \\
& \times (k+m-1) (C_{\alpha,n})^{m} \| \gamma^{(k+m)}_0 \|_{\mathrm{H}^{\alpha}_{k+m}}\\
\le & \dfrac{t^m}{m!} k \cdots (k+m-1)(C_{\alpha,n})^{m} \|\gamma^{(k+m)}_0\|_{\mathrm{H}^{\alpha}_{k+m}}\\
= & \binom{k+ m-1}{m} (C_{\alpha,n}  t )^m \|\gamma^{(k+m)}_0 \|_{\mathrm{H}^{\alpha}_{k+m}}.
\end{split}
\end{equation*}
Then, for $T > 0$ ($T$ will be fixed in the sequel) we obtain
\be
\begin{split}
\|\gamma^{(k)}_{m,t} \|_{C([0,T], \mathrm{H}^{\alpha}_k)} & \le \sum_{j=0}^{m} \|\Xi_{j}^{(k)}\|_{C([0,T], \mathrm{H}^{\alpha}_k)}\\
& \leq  \sum_{j = 0}^m \binom{k+j-1}{j} (C_{\alpha,n}  T )^j\|\gamma^{(k+j)}_{0}\|_{\mathrm{H}^{\alpha}_{k+j}}\\
\end{split}
\ee
Hence, for $0<\eta<1$ (which will be fixed later) one has
\begin{equation}\label{eq:GammaSum}
\begin{split}
\sum_{k=1}^{\8} & \eta^k \|\gamma^{(k)}_{m,t} \|_{C([0,T], \mathrm{H}^{\alpha}_k)}\\
& \leq  \sum_{j = 0}^m \sum_{k=1}^{\8} \eta^k\binom{k+j-1}{j} (C_{\alpha,n}  T )^j\|\gamma^{(k+j)}_{0}\|_{\mathrm{H}^{\alpha}_{k+j}}\\
& \leq \sum_{j = 0}^{\8} \sum_{k=1}^{\8} \eta^k\binom{k+j-1}{j} (C_{\alpha,n}  T )^j\|\gamma^{(k+j)}_{0}\|_{\mathrm{H}^{\alpha}_{k+j}}.\\
\end{split}
\end{equation}
By the direct computation, one has
\be
\begin{split}
\sum_{j=0}^{\infty}\sum_{k=1}^{\infty} & \eta^k\binom{k+j-1}{j} (C_{\alpha,n}  T )^j\|\gamma^{(k+j)}_{0}\|_{\mathrm{H}^{\alpha}_{k+j}}\\
=&\sum_{j=0}^{\infty}\sum_{l=j+1}^{\infty}\eta^{l-j}\binom{l-1}{j}(C_{\alpha,n}  T )^j\|\gamma^{(l)}_{0}\|_{\mathrm{H}^{\alpha}_l}\\
=&\sum_{l=1}^{\infty}\sum_{j=0}^{l-1}\binom{l-1}{j}(C_{\alpha,n}  T/\eta )^j\eta^l\|\gamma^{(l)}_{0}\|_{\mathrm{H}^{\alpha}_l}\\
=&\sum_{l=1}^{\infty}(1+C_{\alpha,n}  T/\eta )^{l-1}\eta^l\|\gamma^{(l)}_{0}\|_{\mathrm{H}^{\alpha}_l}\\
=&\dfrac{\eta}{\eta+C_{\alpha,n}  T}\sum_{l=1}^{\infty}(\eta+C_{\alpha,n}  T)^{l}\|\gamma^{(l)}_{0}\|_{\mathrm{H}^{\alpha}_l}.
\end{split}
\ee
Set $\Gamma_{m,t} = \{ \gamma^{(k)}_{m,t} \}.$ Let $\eta=\xi-C_{\alpha,n}T$ with $0<T<\xi/ C_{\alpha,n}.$ Then, it follows from \eqref{eq:GammaSum} that
\beq\label{eq:SpacetimeEstimate_m-itera}
\| \Gamma_{m,t} \|_{C( [0,T], \mathcal{H}_{\eta}^{\alpha})} \leq \dfrac{\eta}{\xi}\sum_{l=1}^{\infty}\xi ^{l}\|\gamma^{(l)}_{0}\|_{\mathrm{H}^{\alpha}_l}= \dfrac{\eta}{\xi}\|\Gamma_0\|_ {\mathcal{H}_{\xi}^{\alpha}}.
\eeq

Next, we prove that $\{\Gamma_{m,t}\}_{m \ge 1}$ converges to a solution. Indeed, by the above estimates for $\Xi_{j,t}^{(k)}$ we have for any $m,n$ with $n >m$
\be
\|\gamma^{(k)}_{m,t} - \gamma^{(k)}_{n,t}\|_{C([0,T], \mathrm{H}^{\alpha}_k)} \le 2 \sum_{j = m}^{n} \binom{k+j-1}{j} (C_{\alpha,n}  T )^j\|\gamma^{(k+j)}_{0}\|_{\mathrm{H}^{\alpha}_{k+j}}.
\ee
Then,
\be
\| \Gamma_{m,t} - \Gamma_{n,t} \|_{C_{t \in [0,T]} \mathcal{H}_{\eta}^{\alpha}} \le 2 \sum_{j = m}^{n} \sum_{k=1}^{\infty}\eta^k\binom{k+j-1}{j} (C_{\alpha,n}  T )^j\|\gamma^{(k+j)}_{0}\|_{\mathrm{H}^{\alpha}_{k+j}}.
\ee
An immediate computation as above yields that
\be
\| \Gamma_{m,t} - \Gamma_{n,t} \|_{C( [0,T], \mathcal{H}_{\eta}^{\alpha})} \le 2 \dfrac{\eta}{\xi}\sum_{l=m+1}^{\infty}\xi ^{l}\|\gamma^{(l)}_{0}\|_{\mathrm{H}^{\alpha}_l}.
\ee
Since $\Gamma_0 = \{ \gamma^{(k)}_0 \}_{k \ge 1} \in \mathcal{H}^{\alpha}_{\xi},$ it is concluded that $\{  \Gamma_{m,t} \}_{m \ge 1}$ is a Cauchy sequence in $C( [0,T], \mathcal{H}_{\eta}^{\alpha})$ and so converges to some $\Gamma_t \in C([0,T], \mathcal{H}_{\eta}^{\alpha}).$ Taking the limit $m \to \8$ in \eqref{eq:m-iterate} we prove that $\Gamma_t$ is a solution to \eqref{eq:GPHierarchyEqua}. Also, taking $m \to \8$ in \eqref{eq:SpacetimeEstimate_m-itera} yields this solution satisfies \eqref{eq:SpacetimeEstimate-nge1}.

(ii)\; Choose $T = \xi/(5 C)$ and suppose $\Gamma_t, \Gamma'_t \in C( [0,T], \mathcal{H}^{\alpha}_{\eta})$ are two solutions to \eqref{eq:GPHierarchyEqua} with the initial datum $\Gamma_0$ and $\Gamma'_0$ in $\mathcal{H}^{\alpha}_{\xi},$ respectively. Since \eqref{eq:GPHierarchyEqua} is linear, it suffices to consider $\Gamma_t$ instead of $\Gamma_t-\Gamma'_t.$
By \eqref{eq:DuhamelExpan}, for every $m \geq 1$ one has
\begin{equation*}
\begin{split}
\gamma^{(k)}_t = & {\mathcal U}^{(k)}_0(t) \gamma^{(k)}_{0} +
\sum_{j=1}^{m-1}\int_{0}^{t}\mathrm{d}t_1\int_0^{t_1}\mathrm{d}t_2
\cdots \int_0^{t_{j-1}}\mathrm{d}
t_j {\mathcal U}^{(k)}_0(t-t_1) \tilde{B}^{(k)} \cdots\\
& \times {\mathcal U}^{(k+j-1)}_0(t_{j-1}-t_j) \tilde{B}^{(k+j-1)}
{\mathcal U}^{(k+j)}_0(t_j) \gamma^{(k+j)}_0\\
& + \int_{0}^{t} \mathrm{d} t_1 \int_0^{t_1} \mathrm{d} t_2 \cdots
\int_0^{t_{m-1}} \mathrm{d} t_m {\mathcal U}^{(k)}_0 (t-t_1)
\tilde{B}^{(k)} \cdots\\
&\times {\mathcal U}^{(k+m-1)}_0 (t_{m-1}-t_m) \tilde{B}^{(k+m-1)}
\gamma^{(k+m)}_{t_m}\\
\triangleq & \sum_{j=0}^{m-1} \Xi_{j,t}^{(k)} + \tilde{\Xi}_{m,t}^{(k)},
\end{split}
\end{equation*}
with the convention $t_0 =t.$ Note that,
\begin{equation*}
\begin{split}
\| \tilde{\Xi}_{m,t}^{(k)}\|_{\mathrm{H}^{\alpha}_k} \leq & \int_{0}^{t} \mathrm{d} t_1 \int_0^{t_1} \mathrm{d} t_2 \cdots \int_0^{t_{m-1}} \mathrm{d} t_m \Big \| {\mathcal U}^{(k)}_0(t-t_1) \tilde{B}^{(k)} \cdots \\
& \times {\mathcal U}^{(k+m-1 )}_0(t_{m-1}-t_m) \tilde{B}^{(k+m-1)} \gamma^{(k+m)}_{t_m} \Big \|_{\mathrm{H}^{\alpha}_k}\\
\leq & \int_{0}^{t}\mathrm{d}t_1\int_0^{t_1}\mathrm{d}t_2\cdots\int_0^{t_{m-1}}\mathrm{d} t_m~ k \cdots \\
& \quad \times (k+m-1) (C_{\alpha,n})^{m} \| \gamma^{(k+m)}_{t_m} \|_{\mathrm{H}^{\alpha}_{k+m}}\\
\leq &  m (m+1) \cdots (2m-1) (C_{\alpha,n})^{m} \\
& \quad \times \int_{0}^{t}\mathrm{d}t_1\int_0^{t_1}\mathrm{d}t_2\cdots\int_0^{t_{m-1}} \| \gamma^{(k+m)}_{t_m} \|_{\mathrm{H}^{\alpha}_{k+m}} \mathrm{d} t_m .\\
\end{split}
\end{equation*}
Combining this estimate and \eqref{eq:XiEsitmate} yields
\be
\begin{split}
\|\gamma^{(k)}_t \|_{\mathrm{H}^{\alpha}_k}
\leq &  \sum_{j = 0}^{m-1} \binom{k+j-1}{j} (C_{\alpha,n}  T )^j\|\gamma^{(k+j)}_{0}\|_{\mathrm{H}^{\alpha}_{k+j}}\\
& \; + m (m+1) \cdots (2m-1) (C_{\alpha,n})^{m} \\
& \quad \times \int_{0}^{t}\mathrm{d}t_1\int_0^{t_1}\mathrm{d}t_2\cdots\int_0^{t_{m-1}} \| \gamma^{(k+m)}_{t_m} \|_{\mathrm{H}^{\alpha}_{k+m}} \mathrm{d} t_m .
\end{split}
\ee
Then for $m \ge 1$ and $0 < t < T,$ we have
\begin{equation}\label{eq:GammaSum_m-itera}
\begin{split}
\sum_{k=1}^{m} & \eta^k \|\gamma^{(k)}_t \|_{\mathrm{H}^{\alpha}_k}\\
& \leq  \sum_{j = 0}^{m-1} \sum_{k=1}^m \binom{k+j-1}{j} (C_{\alpha,n} T/ \eta )^j \eta^{k+j} \|\gamma^{(k+j)}_{0}\|_{\mathrm{H}^{\alpha}_{k+j}}\\
& \quad + m (m+1) \cdots (2m-1) \Big (\frac{ C_{\alpha,n} }{\eta} \Big )^{m} \\
& \quad \; \times \int_{0}^{t}\mathrm{d}t_1\int_0^{t_1}\mathrm{d}t_2\cdots\int_0^{t_{m-1}} \sum^m_{k=1} \eta^{(k+m)}\| \gamma^{(k+m)}_{t_m} \|_{\mathrm{H}^{\alpha}_{k+m}} \mathrm{d} t_m\\
& \le \dfrac{\eta}{\xi} \|\Gamma_0\|_ {\mathcal{H}_{\xi}^{\alpha}} + \Big (\frac{ C_{\alpha,n} t}{\eta} \Big )^{m} \binom{2m-1}{m} \| \Gamma_t \|_{C( [0,T], \mathcal{H}^{\alpha}_{\eta})}\\
& \leq \dfrac{4}{5} \|\Gamma_0\|_ {\mathcal{H}_{\xi}^{\alpha}} + \binom{2m-1}{m} \Big (\frac{ C_{\alpha,n} T}{\eta} \Big )^{m} \| \Gamma_t \|_{C( [0,T], \mathcal{H}^{\alpha}_{\eta})}.
\end{split}
\end{equation}
By Stirling's formula $m! \approx m^{m+ 1/2} e^{-m}$ we have
\be
\binom{2m-1}{m} \approx \frac{4^m}{\sqrt{m}}.
\ee
Then, taking $m \to \8$ in \eqref{eq:GammaSum_m-itera} yields that
\be
\| \Gamma_t \|_{C( [0,T], \mathcal{H}_{\eta}^{\alpha})} \leq \dfrac{4}{5} \|\Gamma_0\|_ {\mathcal{H}_{\xi}^{\alpha}},
\ee
because $C_{\alpha, n} T/\eta = 1/4.$ Thus, the space-time type bound \eqref{eq:SpacetimeEstimate-Twosolution} holds true.

The uniqueness of the solution in $C( [0,T], \mathcal{H}_{\eta}^{\alpha})$ follows clearly from the estimate \eqref{eq:SpacetimeEstimate-Twosolution}.
\end{proof}

\section{The quintic Gross--Pitaevskii hierarchy}\label{GPHierarchyQuintic}

In this section, we consider the so-called quintic Gross--Pitaevskii hierarchy. Recall that the quintic Gross--Pitaevskii hierarchy is given by
\beq\label{eq:GPHierarchyEquaQuintic}
i \partial_t \gamma^{(k)}_t = \big [- \Delta^{(k)}, \gamma^{(k)}_t \big ] + \mu Q^{(k)} \gamma^{(k+2)}_t,~~\Delta^{(k)} = \sum^k_{j=1} \Delta_{x_j},\; \mu = \pm 1,
\eeq
in $n$ dimensions, for $k \in \mathbb{N},$ where the operator $Q^{(k)}$ is defined by
\be
Q^{(k)} \gamma^{(k+2)}_t = \sum^k_{j=1} \mathrm{t r}_{k+1, k+2} \big [ \delta (x_j - x_{k+1}) \delta (x_j - x_{k+2}), \gamma^{(k+2)}_t \big ].
\ee
It is {\it defocusing} if $\mu =1,$ and {\it focusing} if $\mu= -1.$ We note that the quintic Gross--Pitaevskii hierarchy accounts for $3$-body interactions between the Bose particles (see \cite{CP3} and references therein for details).

\begin{remark}\label{rk:GPQuintic}
In terms of kernel functions we can rewrite \eqref{eq:GPHierarchyEquaQuintic} as follows
\beq\label{eq:GPHierarchyEquaFunctQuintic}
\big ( i \partial_t + \triangle^{(k)}_{\pm} \big ) \gamma^{(k)}_t ({\bf x}_k;{\bf x}'_k) = \mu \big ( Q^{(k)} \gamma^{(k+2)}_t \big ) ({\bf x}_k; {\bf x}'_k ),
\eeq
where, the action of $Q^{(k)}$ on $\gamma^{(k+2)} ({\bf x}_{k+2}, {\bf x}'_{k+2}) \in \mathcal{S} (\mathbb{R}^{(k+2)n} \times \mathbb{R}^{(k+2)n})$ is given by
\be\begin{split}
\big ( Q^{(k)} & \gamma^{(k+2)} \big ) ({\bf x}_k, {\bf x}'_k )\\
& : = \sum^k_{j=1} \big ( Q_{j, k} \gamma^{(k+2)} \big ) ({\bf x}_k, {\bf x}'_k)\\
& : = \sum^k_{j=1} \int d x_{k+1} d x_{k+2} d x'_{k+1} d x'_{k + 2} \gamma^{(k+2)} ({\bf x}_k, x_{k+1}, x_{k+2}; {\bf x}'_k, x'_{k+1}, x'_{k+2}) \\
&\quad \times \Big [ \prod^{k+2}_{\el=k+1} \delta (x_j - x_{\el}) \delta (x_j - x'_{\el}) - \prod^{k+2}_{\el=k+1} \delta (x'_j - x_{\el}) \delta (x'_j - x'_{\el})\Big ].
\end{split}\ee

Let $\varphi \in \mathrm{H}^1(\mathbb{R}^n),$ then one can easily verify that a particular solution to \eqref{eq:GPHierarchyEquaFunctQuintic} with initial conditions
\be
\gamma^{(k)}_{t=0}({\bf x}_k; {\bf x}'_k) = \prod^k_{j=1} \varphi(x_j) \overline{\varphi(x'_j)},\quad k=1,2,\ldots,
\ee
is given by
\be
\gamma^{(k)}_t ({\bf x}_{k}; {\bf x}'_{k} ) = \prod^k_{j=1} \varphi_t (x_j) \overline{\varphi_t ( x'_j )}~~k=1,2,\ldots,
\ee
where $\varphi_t$ satisfies the quintic non-linear Schr\"odinger equation
\beq\label{eq:GPEquaQuintic}
i\partial_t \varphi_t = -\Delta \varphi_t + \mu |\varphi_t|^4 \varphi_t,\quad \varphi_{t=0}=\varphi.
\eeq
\end{remark}

\begin{remark}\label{rk:DuhamelExpansionQuintic}
The Gross-Pitaevskii hierarchy \eqref{eq:GPHierarchyEquaQuintic} can be written in the integral form
\beq\label{eq:GPHierarchyIntEquaQuintic}
\gamma^{(k)}_t = {\mathcal U}^{(k)}_0(t) \gamma^{(k)}_0 + \int^{t}_{0} d s~ {\mathcal U}^{(k)}_0 (t-s) \tilde{Q}^{(k)} \gamma^{(k+2)}_s,\; k=1,2,\ldots,
\eeq
where $\tilde{Q}^{(k)} = - i \mu Q^{(k)}.$ Formally we can expand the solution $\gamma^{(k)}_t$ of \eqref{eq:GPHierarchyIntEquaQuintic}
for any $m \geq 1$ as
\beq\label{eq:DuhamelExpanQuintic}
\begin{split}
\gamma^{(k)}_t & = {\mathcal U}^{(k)}_0 (t) \gamma^{(k)}_0 + \sum^{m-1}_{j=1} \int^t_0 d s_1 \int^{s_1}_0 d s_2 \cdots
\int^{s_{j-1}}_0 d s_j {\mathcal U}^{(k)}_0(t-s_1) \tilde{Q}^{(k)} \cdots\\
& \; \times {\mathcal U}^{(k+j-1)}_0 ( s_{j-1} - s_j ) \tilde{Q}^{(k+j-1)} {\mathcal U}^{(k+2+j)}_0 (s_j) \gamma^{(k+2+j)}_0\\
& \; + \int^t_0 d s_1 \int^{s_1}_0 d s_2 \cdots \int^{s_{m-1}}_0 d s_m {\mathcal U}^{(k)}_0 ( t-s_1 ) \tilde{Q}^{(k)} \cdots \\
& \; \times {\mathcal U}^{(k+m-1)}_0 ( s_{m-1} - s_m ) \tilde{Q}^{(k+m-1)} \gamma^{(k+2+m)}_{s_m},
\end{split}
\eeq
with the convention $s_0 =t.$ 
\end{remark}

Let ${\bf q} = (p_{k+1}, p_{k+2})$ and ${\bf q}' = (p'_{k+1}, p'_{k+2})$ we have
\be\begin{split}
& \big ( Q_{j, k} \gamma^{(k+2)} \big ) ({\bf p}_k, {\bf p}'_k)\\
= & \int d {\bf q} d {\bf q}' \big [ \gamma^{(k+2)} (p_1, \ldots, p_j + p_{k+1}+p_{k+2}- p'_{k+1}- p'_{k+2}, \ldots, p_k, {\bf q}; {\bf p}'_{k+2})\\
& \quad - \gamma^{(k+2)} ({\bf p}_{k+2}; p'_1, \ldots, p'_j +  p'_{k+1}+p'_{k+2}- p_{k+1}- p_{k+2}, \ldots, p'_k, {\bf q}') \big ]
\end{split}\ee
It is proved in \cite{CP3} (Theorem 4.3 there) that for $\alpha > n/2$ there exists a constant $C=C_{n, \alpha}>0$ depending only on $n$ and $\alpha$ such that
\be
\|  Q_{j,k} \gamma^{(k+2)} \|_{\mathrm{H}^{\alpha}_k} \le C \| \gamma^{(k+2)} \|_{\mathrm{H}^{\alpha}_{k+2}},\quad \forall j=1,\ldots, k.
\ee
Then, by slightly repeating the proof of Theorem \ref{th:nge1}, we can obtain the following theorem.

\begin{theorem}\label{th:CauchyProbQuintic}
Assume that $n \geq 1$ and $\alpha > n/2.$ Suppose $\Gamma_0 = \{\gamma^{(k)}_{0}\}_{k \geq 1}\in\mathcal{H}_{\xi}^{\alpha}$ for some $0<\xi<1.$ Then there exists a constant $C = C_{\alpha, n}$ depending only on $n$ and $\alpha$ such that, for a fixed $0 < T < \xi/C$ with $\eta = \xi - C T,$ the following hold.
\begin{enumerate}[{\rm (i)}]

\item There exists a solution $\Gamma_t = \{\gamma^{(k)}_t \}_{k \geq 1} \in C ( [0,T], \mathcal{H}_{\eta}^{\alpha})$ to the Gross-Pitaevskii hierarchy \eqref{eq:GPHierarchyEquaQuintic} with the initial data $\Gamma_0$ satisfying
\begin{equation}\label{eq:SpacetimeEstimate-nge1Quintic}
\| \Gamma_t \|_{C( [0,T],\mathcal{H}_{\eta}^{\alpha})} \leq \frac{1}{\eta \xi}\| \Gamma_0 \|_{\mathcal{H}_{\xi}^{\alpha}}.
\end{equation}

\item For $T = \xi/(5 C),$ if $\Gamma_t$ and $\Gamma'_t$ in $C ( [0,T], \mathcal{H}_{\eta}^{\alpha})$ are two solutions to \eqref{eq:GPHierarchyEquaQuintic} with initial conditions $\Gamma_{t=0} = \Gamma_0$ and $\Gamma'_{t=0} = \Gamma'_0$ in $\mathcal{H}^{\alpha}_{\xi}$ respectively, then
\begin{equation}\label{eq:SpacetimeEstimate-TwosolutionQuintic}
\| \Gamma_t - \Gamma'_t \|_{C( [0,T], \mathcal{H}_{\eta}^{\alpha})} \leq \frac{5}{4 \xi^2} \| \Gamma_0 - \Gamma'_0  \|_{\mathcal{H}_{\xi}^{\alpha}}.
\end{equation}
Consequently, the solution $\Gamma_t$ to the initial problem \eqref{eq:GPHierarchyEquaQuintic} with the initial data in $\mathcal{H}^{\alpha}_{\xi}$ is unique in $C( [0,T], \mathcal{H}_{\eta}^{\alpha})$ for any $0 < T < \xi/ C.$

\end{enumerate}
\end{theorem}

We omit the details of the proof. This result also improves the corresponding one in \cite{CP1} for the regime $\alpha> \frac{n}{2}.$

\

{\it Acknowledgment}\; We are grateful to the anonymous referee for many helpful comments and suggestions, which have been incorporated into this version of the paper.

\end{document}